\documentclass[runningheads]{llncs}

\usepackage[utf8]{inputenc}
\usepackage[T1]{fontenc}
\usepackage[english]{babel}
\usepackage{graphicx}
\usepackage{amssymb}
\usepackage{mathtools}
\usepackage{here}
\usepackage{tabularx}
\usepackage{amsmath}
\usepackage{cite}
\usepackage{academicons}
\usepackage{xcolor}

\usepackage{tikz,xcolor,hyperref}

\definecolor{lime}{HTML}{A6CE39}
\DeclareRobustCommand{\orcidicon}{
	\begin{tikzpicture}
	\draw[lime, fill=lime] (0,0) 
	circle [radius=0.16] 
	node[white] {{\fontfamily{qag}\selectfont \tiny ID}};
	\draw[white, fill=white] (-0.0625,0.095) 
	circle [radius=0.007];
	\end{tikzpicture}
	\hspace{-2mm}
}

\foreach \x in {A, ..., Z}{%
	\expandafter\xdef\csname orcid\x\endcsname{\noexpand\href{https://orcid.org/\csname orcidauthor\x\endcsname}{\noexpand\orcidicon}}
}

\begin{document}

\title{On Hybrid Gene Regulatory Networks}
\author{Adrian Wurm\orcidA{}$^1$ \and Honglu Sun$^2$ }

\authorrunning{A. Wurm}

\institute{BTU Cottbus-Senftenberg, Lehrstuhl Theoretische Informatik,\\ Platz der Deutschen Einheit 1, 03046 Cottbus, Germany,
\url{https://www.b-tu.de/}
\email{wurm@b-tu.de} \and Nantes Université, École Centrale Nantes, CNRS, LS2N, UMR 6004, F-44000 Nantes, France}

\maketitle 
\begin{abstract}
    In this work, we study a class of hybrid dynamical systems called hybrid gene regulatory networks (HGRNs) which was proposed to model gene regulatory networks. 
    In HGRNs, there exist well-behaved trajectories that reach a fixed point or converge to a limit cycle, as well as chaotic trajectories that behave non-periodic or indeterministic. 
    In our work, we investigate these irregular behaviors of HGRNs and present theoretical results about the decidability of the reachability problem, the probability of indeterministic behavior of HGRNs, and chaos especially in 2-dimensional HGRNs.

\end{abstract} 

\section{Introduction}
Hybrid systems are dynamical systems that admit both continuous and discrete behavior. 
They were initially proposed to model the interactions between the discrete digital world and the continuous physical world \cite{alur1991hybrid}.
In recent years, hybrid systems were also used in system biology to study biological networks\cite{farcot2009periodic,chaves2013hierarchy,firippi2020topology}.

In this work, we study a class of hybrid systems called hybrid gene regulatory networks (HGRN) \cite{cornillon2016hybrid,behaegel2016hybrid}, which is an extension of Thomas' discrete modeling framework \cite{thomas1973boolean,thomas1991regulatory}. 
These hybrid systems were proposed to model gene regulatory networks, a framework to describe the regulation influences between genes.

Two widely used formalisms to study gene regulatory networks are discrete models (like Boolean networks \cite{kauffman1969metabolic}) and continuous models (differential equations \cite{barik2010model,almeida2020control} and stochastic models \cite{karlebach2008modelling}).
The major advantage of HGRNs compared to discrete models is that they can provide and process temporal information of the systems.
Compared to purely continuous models, the dynamical properties of HGRNs are easier to analyze, which makes them a more reasonable choice to apply on large and complex instances.

HGRNs are similar to piecewise-constant derivative systems (PCD systems) \cite{PCD} which are a special case of hybrid automata \cite{alur1993hybrid}. The major difference between HGRNs and PCD systems, as for example in \cite{PCD,asarin2012low,sandler2019deciding}, is the existence of sliding modes, which means that when a trajectory reaches a black wall (a boundary of the discrete state which can be reached but cannot be crossed by trajectories), the trajectory does not end there but is instead forced to move along this black wall.
There exist other methods to define the behavior of trajectories on a black wall \cite{gouze2002class,plahte2005analysis} which are different from the sliding modes in HGRNs.

One important research topic of dynamical models of gene regulatory networks is how to analyze the reachability of states, for instance, whether a trajectory starting in a certain state can reach another given state.
These questions arise in cell reprogramming \cite{mandon2019sequential}. In \cite{honglu}, the reachability problem of HGRNs is investigated and methods are proposed to analyze the reachability of regular trajectories, meaning those that reach a fixed point or converge to a limit cycle. 

In HGRNs however, there also exists so-called chaotic trajectories that cross an infinite sequence of discrete regions so that this sequence does not end up in a cycle.
For now, our understanding of this irregular behaviour is quite limited.
Further understanding of these trajectories can potentially lead to the development of new methods to analyze the reachability problem for them. 
In this work, we want to fill the gaps in this research field by investigating irregular behavior occuring in HGRNs. 
Our major contributions are:
\begin{itemize}
    \item We prove that the reachability problem for HGRNs is undecidable and that this enforces the existence of chaotic trajectories.
    \item We prove that under reasonable conditions, the probability for a random trajectory to admit indeterminism tends to 1 as the dimension of the system tends to infinity.
    \item We argue that  in two-dimensional HGRNs, chaos is not possible without indeterminism, and conclude that the probability for a random two-dimensional HGRN to admit chaos with indeterminism is 0. 
\end{itemize}

The paper is organized as follows:
In Section \ref{Section:preliminaries} we collect basic notions and recall the definition of HGRNs as well as their properties. In Section \ref{Section:structure} we prove that HGRNs are Turing-powerful and conclude that chaotic trajectories must exist. In Section \ref{Section:Indeterminism} we give some structural results on indeterminism and provide conditions as well as probability estimations for the bifurcation of trajectories.

The paper ends with some open questions. 

\section{Preliminaries and Network Reachability Problems}\label{Section:preliminaries}

We start by defining the problems we are interested in; here,
we follow the definitions and notions of \cite{honglu} for everything related to 
hybrid gene regulatory networks. 

Intuitively, a HGRN is a  an $N$-dimensional rectangular cuboid/hyperrectangle with integer edge lengths decomposed into hypercubes with unit edge length. Each such hypercube $d$ is called a discrete state and is associated with a celerity $c_d$ enforcing the direction in which a trajectory is allowed to move inside $d$, so a trajectory may only change its direction in the moment of transition between two discrete states.

\begin{definition}[Hybrid gene regulatory network (HGRN)]
Let $N\in\mathbb N$ and let $n_i\in\mathbb N$ for each $i\in\{1,...,N\}$. A \emph{Hybrid gene regulatory network} (HGRN) on $N$ genes with $n_i$ levels is a pair $\mathcal{H}=(E, c)$, where
\[E := \left\{ d \in \mathbb{N}^N \mid \forall i \in \left\{1,2,...,N \right\} : 0\leq d_i \leq n_i \right\}\]

is called the set of discrete states, and $c:E\rightarrow\mathbb{R}^{N}$ is called the celerity function.
\end{definition}

\begin{definition}[Hybrid state of HGRN]
A \emph{hybrid state} of a HGRN is a pair $h=(\pi,d)$ containing a \emph{fractional part} $\pi\in[0,1]^{N}$, and a discrete state $d\in E$.
$E_h$ is the set of all hybrid states and $E_{sh}$ is the set of all finite or infinite sequences of hybrid states: 
\[E_{sh} = \left\{ (h_0,h_1,...,h_m) \in (E_h)^{m+1} \mid m \in \mathbb{N}\cup \{\infty\} \right\}.\]
Note that $E_h=:E_h^1\subseteq E_{sh}$.
\end{definition}

\begin{definition}[Boundary]
A boundary of a discrete state $d$ is a set of hybrid states defined by $e(i,\pi_0,d) := \left\{ (\pi,d) \in E_h \mid \pi^{i} = \pi_0, \right\}$, where $ i \in \left\{ 1, 2,..., N\right\}$ and $\pi_0 \in \left\{0,1\right\}$.
\end{definition}

\begin{definition}[Hybrid trajectory]
 Let $t_0 \in \mathbb{R}^+ \cup \{\infty\}$ and $[0,t_{0}]$ a time interval. For a function $\tau:[0,t_{0}]\rightarrow E_{sh}$, we denote by abuse of notation by $\pi:[0,t_{0}]\rightarrow E_{sh}$ the projection of $\tau$ on the second component.

 For $\tau(t)=h$, we denote by  $\frac{d\tau(t^+)}{dt}:=\frac{d\pi(t^+)}{dt}$ the temporal derivative of the fractional part of $\tau$, $\frac{d\tau(t^-)}{dt}$ and $\frac{d\tau(t)}{dt}$ are defined analogous.
 
 We call $\tau$ a hybrid trajectory, if for every time $t\in[0,t_{0}]$ with $\tau(t)=h$ and $h=(h_0,...,h_m)\in E_h^{m+1}$ a sequence of $m+1$ hybrid states $h_j=(\pi_j,d_j)$, the following conditions are satisfied:

\begin{itemize}
\item For $t=0$ we have $m=0$, and $\tau(0)\in E_h$ is not on any boundary.
\item If $h_0$ does not belong to any boundary, then $m=0$ and $\frac{d\tau(t)}{dt} = c(d_0)$ for $t>0$ and $\frac{d\tau(t^+)}{dt} = c(d_0)$ for $t=0$ 
\item If $h_0$ only belongs to one boundary $e$, let $e$ be the upper boundary in the $i^{th}$ dimension (the following is analogous when $e$ is the lower boundary). In case $d_0^i$ is not the maximal discrete level of gene $i$, let the discrete state on the other side of $e$ be noted as $d_{r}$, where $d_{0}^{k} = d_{r}^{k}$ for all $k \neq i$, and $d_{0}^{i}+1 = d_{r}^{i}$.
There are four possible cases:
	\begin{itemize}

	\item If $c(d_{0})^i = 0$, then $m=0$ and the trajectory from the current state will slide along the boundary $e$ with speed $\frac{d\tau(t^{+})}{dt} = \frac{d\tau(t^{-})}{dt} = c_{0}$, which is then called a \emph{neutral boundary} of $d_{0}$. 
	\item If $c(d_{0})^i > 0$, and either $d_0^i$ is the maximal discrete level of the $i^{th}$ gene, or $d_0^i$ is not the maximal discrete level of the $i^{th}$ gene but the $i^{th}$ component of $c(d_{r})$ is negative, then $m=0$ and the trajectory from the current state will slide along the boundary $e$, which is called an \emph{attractive boundary} of $d_{0}$. We then have that $\frac{d\tau(t)}{dt}^k = c(d_{0})^k$ for all $k \neq i$, $\frac{d\tau(t^{+})}{dt}^i = 0$ and $\frac{d\tau(t^{-})}{dt}^i \in \{0,c(d_{0})^i\}$.
	\item If $c(d_0)^i > 0$, $d_0^i$ is not the maximal discrete level of the $i^{th}$ gene, and the $i^{th}$ component of $c_{r}$ is positive, then the trajectory from the current state instantly crosses  the boundary $e$ and enter the discrete state $d_{r}$. We then have $m=1$ and $\tau(t)=((\pi_0,d_0),(\pi_1,d_1))$, where $d_1=d_r$, $\pi_0^j=\pi_1^{j}$ for all if $j \neq i$, and $\pi_1^{i} = 0$. Such a boundary is called an \emph{output boundary} and the counterpart of $d_{0}$ is called  an \emph{input boundary} of $d_{r}$. We then demand that $\frac{d\tau(t^{+})}{dt} = c(d_{r})$ and  $\frac{d\tau(t^{-})}{dt} = c(d_{0})$.
	\end{itemize}
\item If $h$ belongs to several boundaries, then the previous cases can be mixed:
	\begin{itemize}
	\item If in these boundaries there is no output boundary, then the trajectory from the current state will exit all input boundaries and slide along all attractive or neutral boundaries.
	\item If in these boundaries there is only one output boundary, then the trajectory from the current state will cross this output boundary.
	\item If in these boundaries there are several output boundaries, then the trajectory from the current state can cross any of them, but can only cross one boundary at one time, which causes indeterministic behavior.
	\end{itemize}
\end{itemize}

A hybrid trajectory $\tau$ is called \begin{itemize}
    \item \emph{finite}, if $t_0\neq\infty$ and the sequence of states that are reached is finite,
	\item \emph{periodic}, if it is defined on $t_0=\infty$ and $\exists T, t_1>0, \forall t \in [t_1, \infty[, \tau(t) = \tau(t+T)$, 
    \item \emph{chaotic}, if the sequence of discrete states that it passes is infinite and non-periodic, meaning it does not end up in a cycle and
    \item \emph{regular}, if it is not chaotic. 
\end{itemize}

\end{definition}

\begin{definition}[Reachability Problem]
The reachability problem for HGRNs is the following: Given a HGRN and two sets $S,T$ of hybrid states in form of linear inequalities, the question is whether there exists a trajectory $\tau$ starting in a hybrid state $s\in S$ and reaching at some time $t$ a state $f\in F$, meaning 
\[\exists \tau \exists s\in S,f\in F, t>0 : \tau(0)=s \land \tau(t)=f.\]
If such a trajectory exists, we say that $F$ is reachable from $T$.
\end{definition}
Note the special cases that $S,T$ might contain only one hybrid state or that they may be a whole discrete state. It was shown in \cite{honglu} that the reachability problem is decidable for regular trajectories.

\begin{definition}[Area of attraction]
Consider a subset $A$ of hybrid states of a HGRN. The area of attraction $\varphi(A)$ is the set of all hybrid states with the property that every trajectory starting in them will eventually intersect $A$. The extended area of attraction $\Phi(A)$ is the set of all hybrid states with the property that there exists a trajectory starting in them that will at some point intersect $A$. 
\end{definition}

\begin{example}
 Fig. \ref{fig:binary} is a sketch of a HGRN on three genes with levels $n_1=2,n_2=4,n_3=2$, the arrows indicate the celerities where an encircled circle indicates an upward movement and an encircled X a downward movement, missing arrows indicate that the celerity is the zero-vector. The left side of the figure illustrates level 1 of gene 3, the right side level 2.
 \newpage
\begin{figure}[h]
\centering
\includegraphics[width=0.75\textwidth]{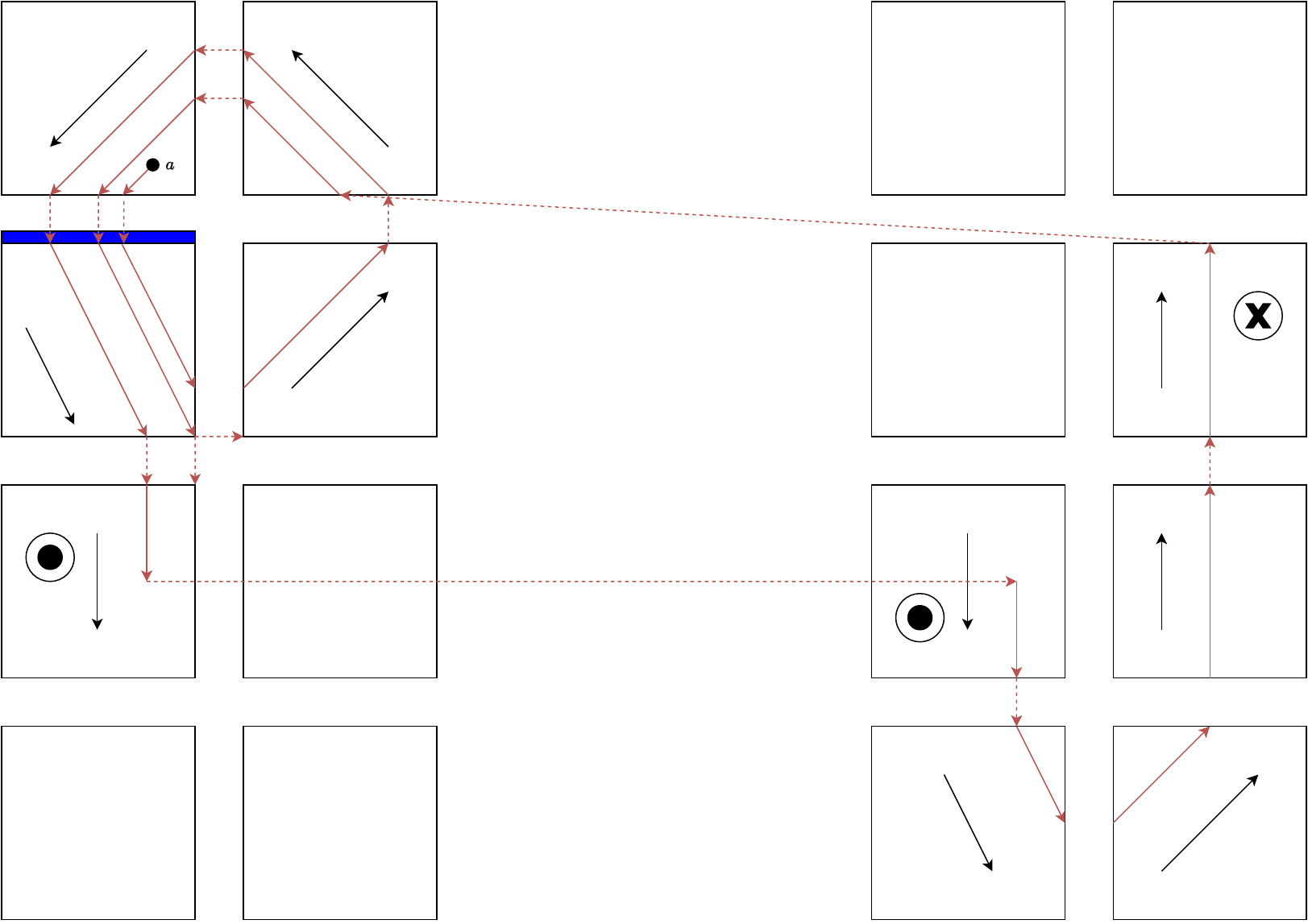}
 \caption{Binary expansion}
  \label{fig:binary}
\end{figure}

 The network computes the binary expansion of a number in $[0,1]$ in the following way: If a trajectory passes $A=e(2,1,(1,2,1))$ (blue area), interpret its distance to the rightmost part $A\cap B , B=e(1,1,(1,2,1))$ as its current value $x$. If $x<\frac12$, the trajectory will enter the small circle of discrete states $(1,2,1),(2,2,1),(2,1,1)$ and $(1,1,1)$ and finally end up in $A$ again with value $2x$. If $x>\frac12$, it will pass the big circle of  discrete states $(1,2,1)$ (where $\frac12$ is subtracted from the value)$,(1,3,1),(1,3,2),(1,4,2)$ (where its value is multiplied with two) $,(2,4,2),(2,3,2),(2,2,2),(2,1,1)$ and $(1,1,1)$ and again reach $A$, this time with value $2(x-\frac12)$. If one follows the trajectory and interprets the first case as 0 and the second case as 1, one obtains the binary expansion of the initial $x$.

 If for example we start in $a=((\frac 3{16},\frac 3{16}),(1,1,1))$, we obtain the trajectory drawn in red. It passes $A$ at value $\frac 3{8}$, so it runs through the small circle and enters $A$ again at value $\frac34$, therefore entering the big circle ending up in $A$ again at value $\frac12$. So far, we obtained the binary expansion 0.01, but now the trajectory becomes indeterministic because the celerities allow a bifurcation at $((1,1),(1,2,1))$. This corresponds to the fact that $\frac 3{8}$ has the two binary expansions $0.011$ and $0.010\bar1$.

 Observe that we obtain a periodic trajectory for every rational value $x$ and a chaotic trajectory for $x$ irrational.
\end{example}

\section{Turing completeness and chaos}\label{Section:structure}

Turing-completeness was already shown for PCD systems in \cite{PCD}, a model related to HGRNs. We would now like to prove it for HGRNs, for they are more restrictive in the sense that they only allow a change of derivatives at integer coordinate changes and admit indeterminism with positive probability. 

\begin{theorem}\label{TC}
HGRNs are Turing powerful.
\end{theorem}

\begin{proof}
Let $M=(Q,\Sigma,\Gamma,\delta,q_1,\square,F)$ be a deterministic Turing-Machine starting on an empty tape, $Q=\{q_1,...,q_n\} , F=\{q_n\}$. Assume that all symbols are encoded in binary, so $\Sigma=\Gamma=\{0,1\}$ and that the encoding of $\square$ consists only of zeroes. Also assume that for each two states $q_i$ and $q_j$ there are transitions in only one direction, either $q_i\rightarrow q_j$ or $q_j\rightarrow q_i$ but not both, a property that can be guaranteed by introducing auxiliary states. We construct an equivalent HGRN $H$ with the following genes:

\begin{description}
\item[$\bullet$] One with two discrete levels storing the current value called $G_c$ 
\item[$\bullet$] $G_{\ell,1}$ with three discrete levels and $G_{\ell,2}$ with two discrete levels for the tape values left of the head, where the current value is stored as the binary expansion with the most significant bits corresponding to the values closest to the head in the fractional part of the upmost discrete state of $G_{\ell,1}$.
\item[$\bullet$] $G_{r,1}$ and $G_{r,2}$ analogous for the tape values right of the head.
\item[$\bullet$] One called $G_i$ for every $q_i\in Q$ with two discrete levels each.
\item[$\bullet$] Two called $G_\alpha$ and $G_\beta$ with two discrete levels each, keeping track of the current calculation step as in Fig. \ref{fig:Tcycle}.
\end{description}
The whole network is therefore of shape $[0,2]\times [[0,3]\times[0,2]]^2\times[0,2]^n\times[0,2]^2$ and our trajectory starts in $((0,0),[(2,0),(0,0)]^2,(1,1),[(0,0)]^{n-1},(0,1),(0,0))$.

Let $(q_i,a)\rightarrow(q_j,b,d)$ be a transition in $\delta$, without loss of generality $d=r$ stating that the head moves to the right ($d=\ell$ is analogous and $d=n$ is easily derived). We choose the celerities as follows:

\begin{figure}[h]
\centering
\includegraphics[width=0.2\textwidth]{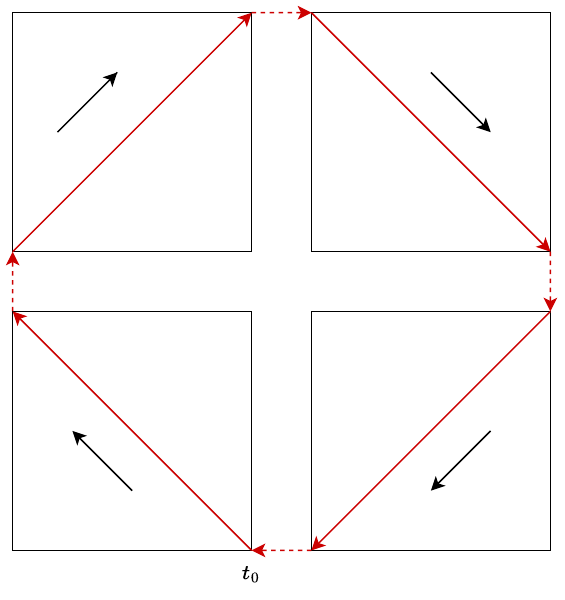}
 \caption{Calculation step accounting}
  \label{fig:Tcycle}
\end{figure}

In the first calculation step $\alpha=\beta=0$ we put $q_j$ into the upmost position. Note that we do not loose the information on which state we started in because we assumed that there is no transition from $q_j$ to $q_i$.

The celerities describing how to update the left side tape start with making the bits less significant by one position because the head recedes by one position, which corresponds to dividing the numerical value of $G_{\ell,1}$ by two, see Fig \ref{fig:Turgadle}, left. This is done in the first calculation step $\alpha=\beta=0$ along with adding half of $b$ to it, see Fig. \ref{fig:Turgadle}, middle, because $b$ is on the position left of the head after that transition.

In the second calculation step $\alpha=0,\beta=1$, this value is then copied back from $G_{\ell,2}$ to the original position, see Fig. \ref{fig:Turgadle}, right, $G_{\ell,1}$, both $G_{\ell,1}$ and $G_{\ell,2}$ remain inactive for the last two steps.

\begin{figure}[h]
\centering
\includegraphics[width=0.5\textwidth]{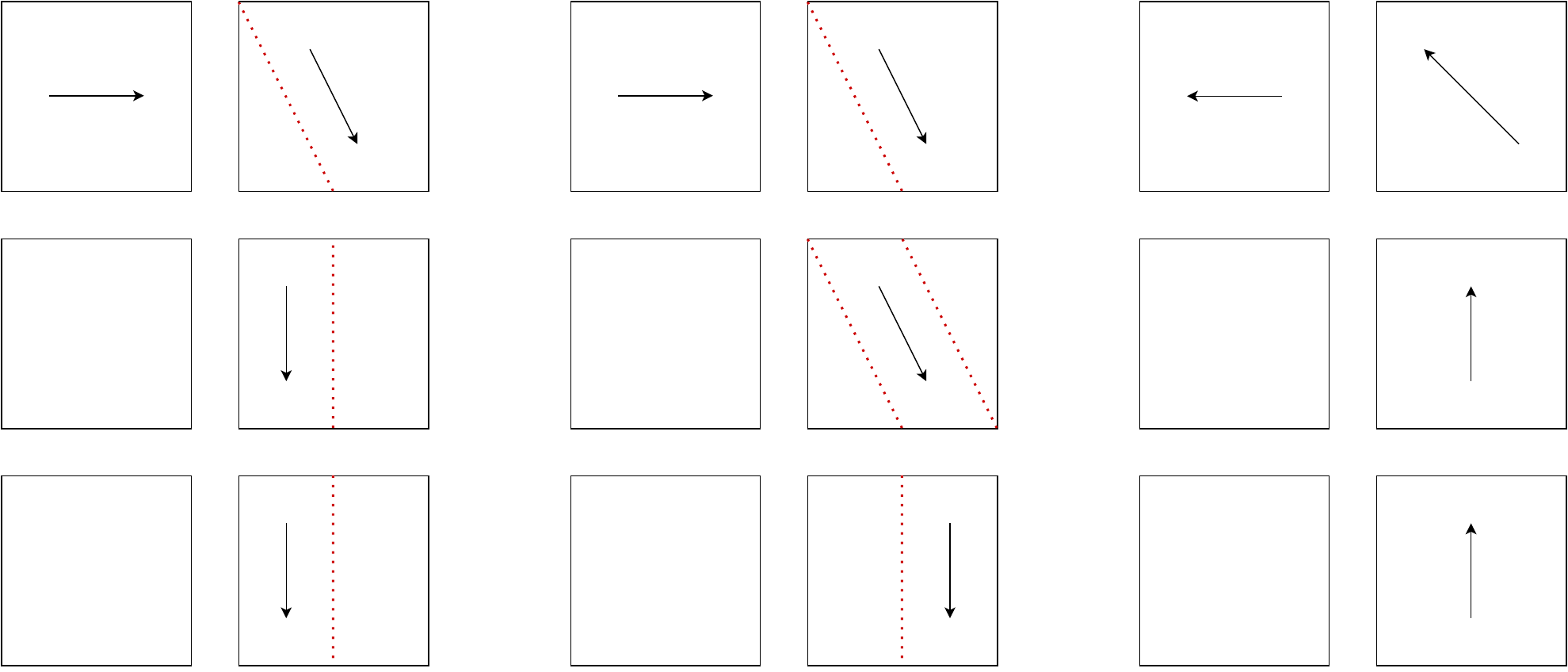}
 \caption{Left side tape update}
  \label{fig:Turgadle}
\end{figure}

The celerities describing how to update the right side tape work similar. We start by multiplying the value by two, see Fig. \ref{fig:Turgadri} left, because each bit now becomes more significant by one position. This is done in the first calculation step $\alpha=\beta=0$, where we also catch the trajectory on the border between $G_{r,1}=1$ and $G_{r,1}=2$ if it is in $G_{r,2}=1$ and let it pass until the lower boundary of $G_{r,1}=1$ if $G_{r,1}=2$, in order to keep track of what our new current value $G_c$ will be in the next Turing step.

\begin{figure}[h]
\centering
\includegraphics[width=0.5\textwidth]{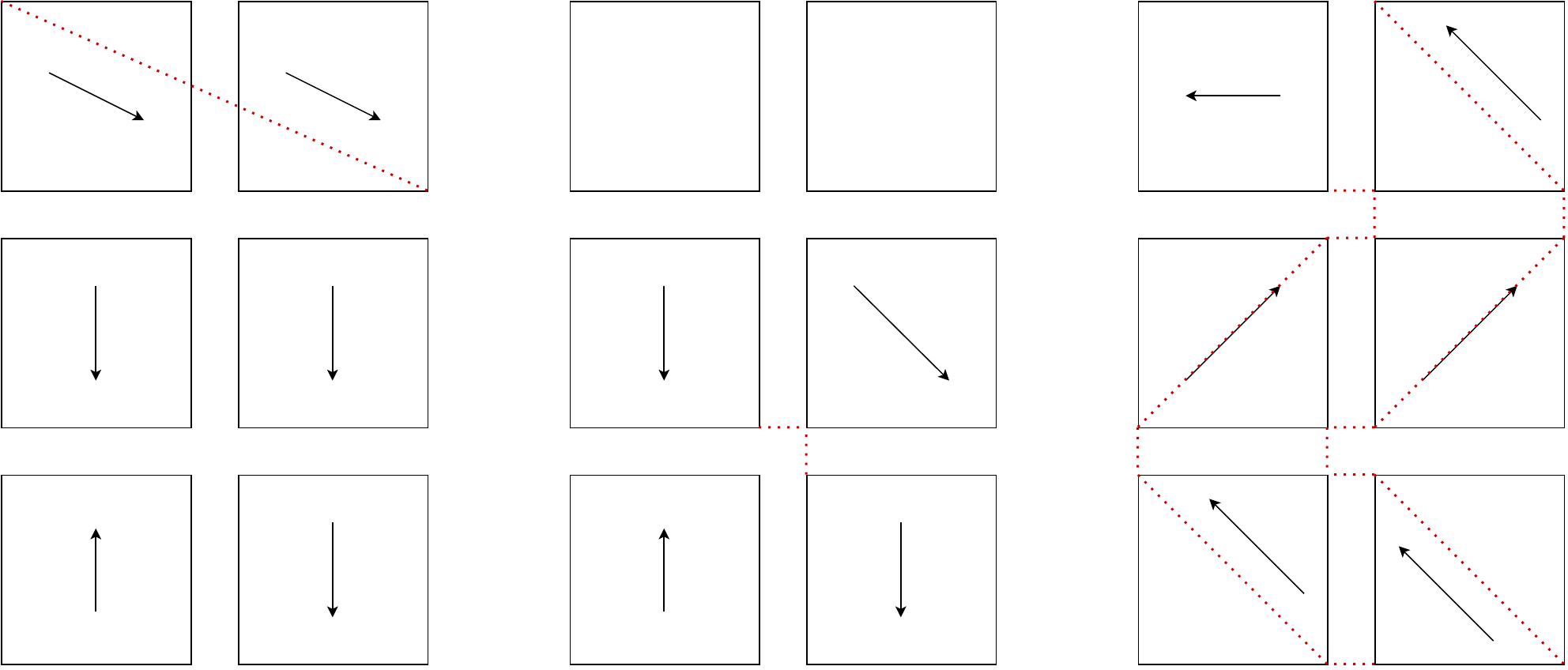}
 \caption{Right side tape update}
  \label{fig:Turgadri}
\end{figure}

Note that if the initial value was $\frac12$, we have an indeterminism when hitting the lower boundary of $G_{r,1}=3$ and the upper boundary of $G_{r,2}=1$. In the second calculation step $\alpha=0,\beta=1$, we therefore adjust the celerities to make sure that we end in $G_{r,2}=2$ whenever the initial value was greater than or equal $\frac 12$, see Fig. \ref{fig:Turgadri} middle (the case that all positions on the right side tape are 1's is not possible).

In the third and fourth calculation step $\alpha=1,\beta=1$ and $\alpha=1,\beta=0$, we pop the most significant bit of the right side tape to be the new $G_c$, which means that if the value of $G_{r,1}$ is greater than 1 we subtract 1 from $G_{r,1}$ and add move $G_c$ to the upmost position, otherwise to the lowest, see Fig. \ref{fig:Turgadri} right. 


Note that we loose the information on whether we moved left or right, so we actually have the celerities of the fourth step on the right side of the tape for the left side of the tape as well. However, the left side is not effected by that because it stays in  $G_{\ell,2}=3$ from the end of step two on and therefore never enters a hypercube with non-zero celerities induced be the update of $G_c$. 

In the fourth step we also bring $G_i$ to the lowest possible point. By assumption of anti-symmetry, $i$ can still be distinguished from $j$ at that point without the knowledge of $d$.$\hfill\blacksquare$
\end{proof}

\begin{corollary}
The reachability problem for HGRNs is undecidable.
\end{corollary}

\begin{proof}
The halting problem can be reformulated as the reachability problem of state $q_n$.$\hfill\blacksquare$
\end{proof}

Note that the trajectories that are obtained by the Turing machines that behave chaotic automatically lead to chaotic trajectories, for the sequence of reached states cannot have a repeating pattern.  

\section{Indeterminism}\label{Section:Indeterminism}
Most attempts to solve verification tasks such as the reachability problem become a lot more difficult as soon as trajectories are not regular any more. It is therefore desirable to estimate how likely a trajectory is to be chaotic or indeterministic.
\begin{example}\label{probchaos}
In opposition to what is claimed in \cite{honglu}, indeterminism will happen with non-zero probability in networks with at least 3 genes having at least two discrete levels each when all celerities and the starting point are chosen uniformly random in $(0,1)$. Fig. \ref{fig:Chaos} is a sketch of a HGRN with 3 Genes $G_1,G_2$ and $G_3$ with two levels 0 and 1 each that allows the indeterministic red trajectory starting in point $a$ in $(0,0,0)$.

\begin{figure}[h]
\centering
\includegraphics[width=0.5\textwidth]{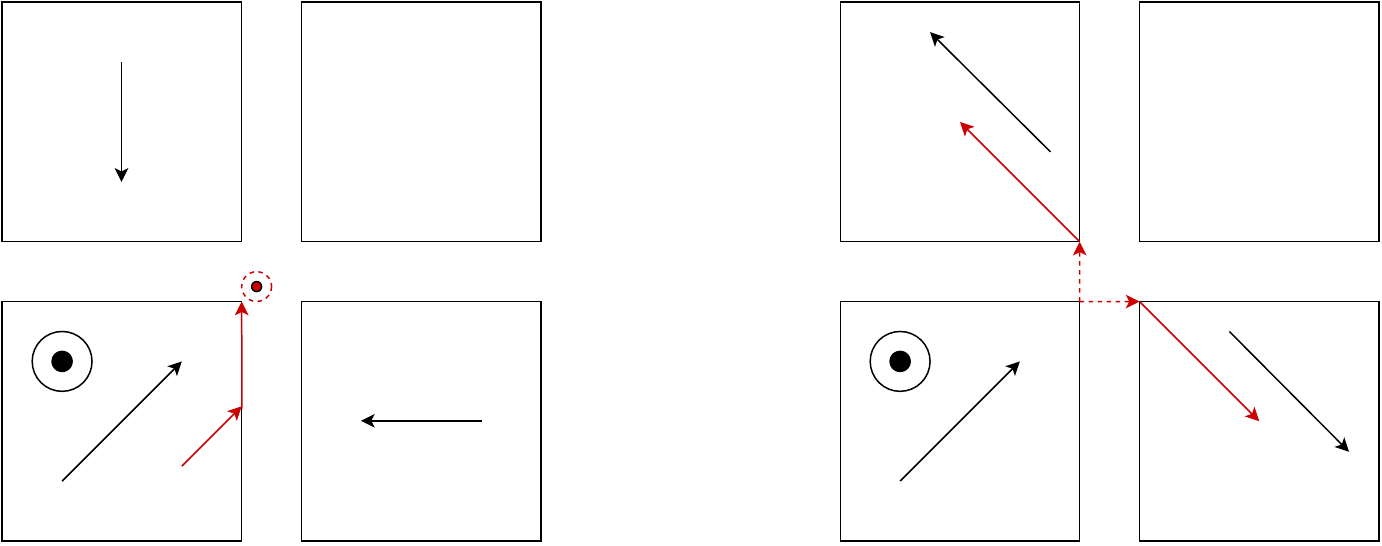}
 \caption{Indeterminism with positive probability}
  \label{fig:Chaos}
\end{figure}

The trajectory reaches the attractive boundary to discrete state 1 of $G_1$, slides along until it reaches the attractive boundary to discrete state 1 of $G_2$ and slides "upwards" along this intersection of discrete domains until it reaches the hybrid state $((1,1,1),(0,0,0))$. The border to the discrete state 1 of $G_3$ is an output boundary, so it is crossed and the hybrid state $((1,1,0),(0,0,1))$ is immediately reached. Now the boundaries to the discrete levels 1 of $G_1$ and $G_2$ both are output boundaries, so the trajectory may pass any of them, leading to indeterminism. Note that we did not rely on probability-zero-events; The probability of the celerity in the first hypercube to point into the middle in all three dimensions is $(\frac12)^3$, the first two reached boundaries are then attractive with probability $\frac 12$ each as well as the third one being and output boundary and the two new boundaries after the crossing are output boundaries with probability $\frac14$ each. The accumulated probability is then $\frac18\cdot(\frac12)^2\cdot\frac12\cdot\frac14^2=(\frac12)^{10}$. This probability significantly rises once we increase the size of the network, meaning the number of genes as well as their dimensions.
\end{example}

We can be more general:

\begin{theorem}\label{PC}
Let all celerities of a HGRN with $n$ genes and at least two discrete levels per gene as well as the starting point of a trajectory be randomly chosen by the uniform distribution in $[0,1]^n$.
Then the probability for this trajectory to lead to indeterminism tends to 1 as $n$ tends to $\infty$. 
\end{theorem}

\begin{proof}
First, note that for each gene we have that an inner boundary is reached at some point with probability at least $\frac12$. These events are independent, so it follows by the central limit theorem that the probability of at least one fourth of all genes reaching an inner boundary tends to 1. Each of these $\frac n4$ boundaries is attractive with probability $\frac12$, if we group them into $\frac n{12}$ triples of boundaries, each triple starts with two attractive boundaries followed by an output boundary with probability $\frac18$ as in Example \ref{probchaos}. The probability that after the crossing of such a third boundary, the two other boundaries become output boundaries is $(\frac12)^2$ each, so the accumulated probability for each triple to lead to indeterminism is at least $\frac18\cdot(\frac12)^2\cdot(\frac12)^2=(\frac12)^7$. The probability that none of the $\frac n{12}$ triples leads to indeterminism is therefore at most $(\frac{2^7-1}{2^7})^{\frac n{12}} \xrightarrow{n\rightarrow\infty} 0$.$\hfill\blacksquare$
\end{proof}

Note that Example \ref{probchaos} relied on the third dimension, for hitting two edges at the same time without sliding before is highly unlikely:

\begin{theorem}\label{tdnc}
Random HGRNs in two dimensions admit chaos only with probability 0.
\end{theorem}

\begin{proof}
Since a trajectory cannot cross itself, it is obvious that chaos can only be obtained by bifurcation.

A trajectory starts in the inner of a hypercube with probability 1. In order to bifurcate, it needs to reach two output boundaries at the same time, meaning it has to hit one of finitely many points. Reaching such an intersection point $q$ by sliding towards it along an attractive boundary will not lead to indeterminism at $q$ because the boundary which one is sliding along is not an output boundary. A point $p$ in a two-dimensional HGRN that has an area of attraction with non-zero measure must either have integer coordinates connected to an attractive boundary or be reachable from such a point, which is a zero-set.

The area of attraction of a point that admits two output boundaries is therefore one-dimensional with probability 1, meaning it is a zero-set.  $\hfill\blacksquare$

\end{proof}

Note that an indeterministic network such as in Example \ref{probchaos} inherently enables chaotic trajectories if both choices lead back into the area of attraction of the point $a$ at which the path was bifurcating, therefore forming two cycles $T_1,T_2$ with the same starting point $a$.

\begin{figure}[h]
\centering
\includegraphics[width=0.25\textwidth]{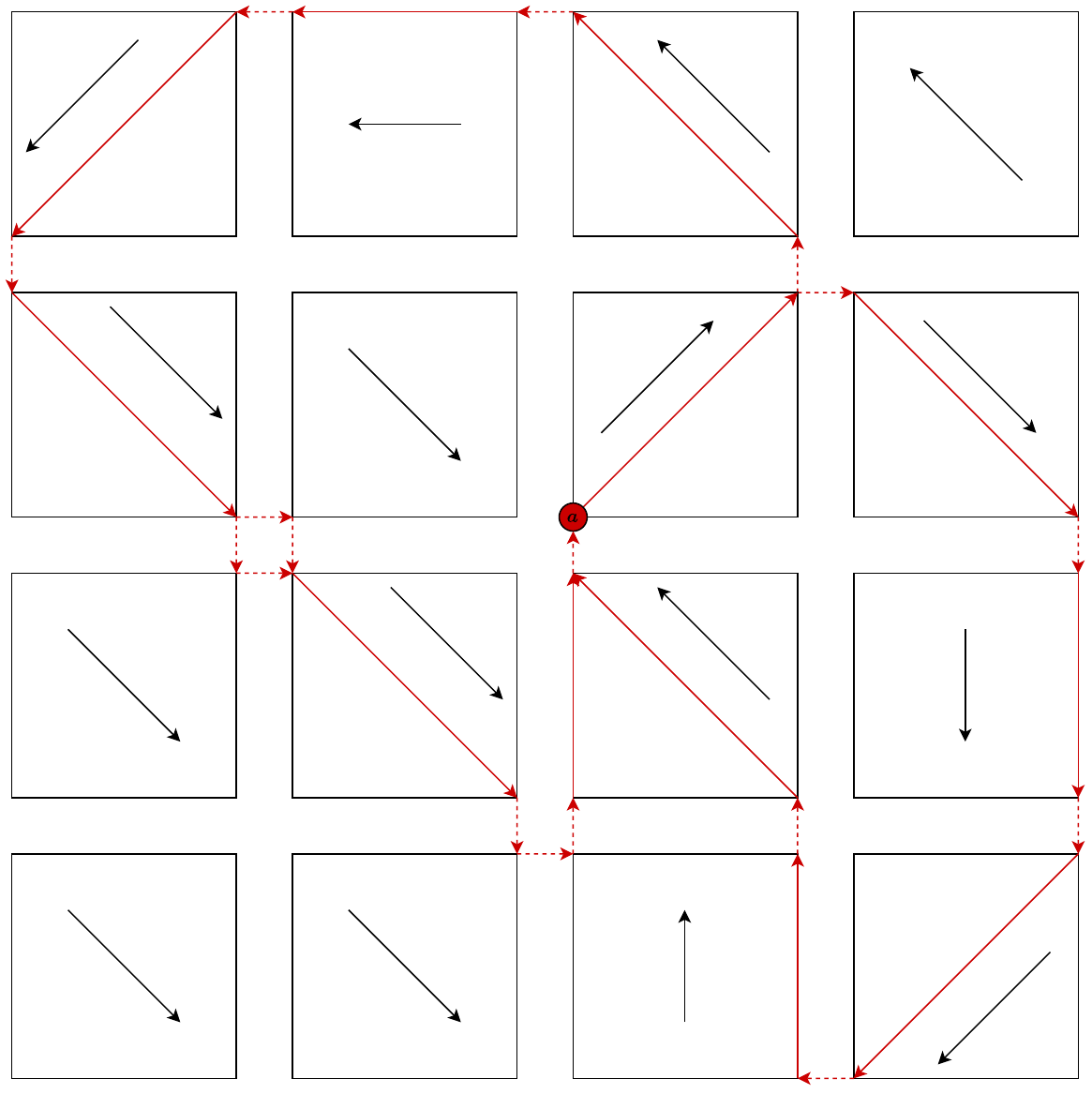}
 \caption{Guaranteed indeterminism with positive probability}
  \label{fig:Guacyc}
\end{figure}

 The idea is sketched in Fig. \ref{fig:Guacyc} for two dimensions.

Take an irrational number, for example $\pi$, and when reaching $a$ for the $n$-th time, follow the cycle $T_i$ iff the $n$-th position in the binary expansion of $\pi$ is $i$. The picture shows indeterministic behavior in two dimensions, which only happens with probability zero, it is however obvious that this construction can be lifted into three dimensions using Example \ref{probchaos} as a gadget to obtain a similar chaotic path with positive probability. 

Fig. \ref{fig:Guacyc} is also an example for a two-dimensional HGRN that leads to indeterminism independent by the starting point of the trajectory, and therefore to chaos, because every possible trajectory in $\mathcal N$ will at some point pass $a$, meaning $\varphi(a)=\Phi(a)=\mathcal N$.

This kind of chaos differs from the chaos that makes HGRNs Turing-complete. It depends on the choices made on the way rather than the restrictions of the network, and another trajectory starting in the same point might be periodic. To examine more precisely how the reachability problem and the identification of chaos relate to those, from now on we call a chaotic path that reaches a certain point an infinite number of times recurrent, and otherwise say that its chaos is proper.

The proof of Theorem \ref{TC} relies on proper chaos, so reachability could still be decidable in recurrent chaos.

The following is a direct consequence of \cite{PCD}:
\begin{lemma}
A two-dimensional HGRN never admits proper chaos.
\end{lemma}

We now want to define a property that most HGRNs of natural origin share:

\begin{definition}
For a HGRN $\mathcal N$ with two discrete levels for every gene and $H=(V,E,\rho)$ a complete weighted directed graph with $V=\{1,...,N\}$ the genes of $\mathcal N$ and $\rho:E\rightarrow\mathbb Q$, we say that $\mathcal N$ is induced by $G$, iff for every $\alpha\in\{0,1\}^N$ we have that celerities are the weighted sums over the incoming edges
\[c(\alpha)_j=\sum\limits_{\substack{i=1 \\ \alpha_i=1}}^N\rho(i,j)-\sum\limits_{\substack{i=1 \\ \alpha_i=0}}^N\rho(i,j)\]
We call $H$ the influence graph of $\mathcal N$.
\end{definition}

We want to pose the question whether our problems become easier when inducedness is given.

Note that induced HGRNs form a zero-set among all possible HGRNs.

\begin{theorem}\label{wei}
It can be decided whether a HGRN is induced or not in P. Moreover, the influence graph is unique and can be computed in polynomial time.
\end{theorem}

\begin{proof}
When given all celerity vectors, the decisional as well as the computational problem are linear equation systems on rational domains and therefore solved efficiently. It remains to show that there can only be one solution.

Suppose there were two graphs $H_1$ and $H_2$ inducing the same HGRN, then for all dimensions $i,j\in\{1,...,n\}$ we have that 

\[\rho_1(i,j)=\frac{c([0]^{i-1},1,[0]^{n-i})_j-c(0,...,0)_j}2=\rho_2(i,j)\]
which gives us uniqueness.$\hfill\blacksquare$
\end{proof}

The constellation in Example \ref{probchaos} can be obtained by an induced HGRN with positive probability, for example by

\begin{center}
\includegraphics[width=0.5\textwidth]{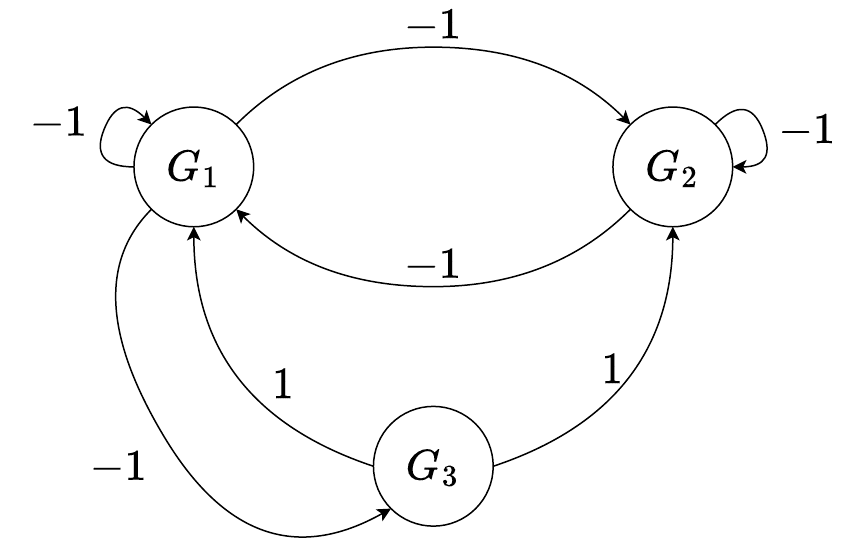}
\end{center}

\section{Conclusion and Further Questions}
We examined the reachability problem for HGRNs and showed in what way the size and dimension of the net influence chaos and indeterminism.

Further open questions are:

\begin{description}
\item[1.)] Are induced HGRNs still Turing-powerful?
\item[2.)] Does Theorem \ref{PC} still hold for induced HGRNs?
\item[3.)] For $H_\mathcal N$ the set of all possible trajectories in $\mathcal N$, how do the sets $R_t:=\{h(t)\mid h\in H_\mathcal N\}$ behave for large $t$? With what probability and how fast do they approach a zero-set?
\item[4.)] Can these sets $R_t$ be computed efficiently?
\item[5.)] Can HGRNs be expressed by induced HGRNs with the help of auxiliary dimensions?
\end{description}
\textbf{Acknowledgment}: We want to thank Klaus Meer for helpful discussion.

\bibliography{literatur}

\bibliographystyle{plain}

\end{document}